\documentclass{article}

\usepackage{amsthm}
\usepackage{graphics,graphicx}

\begin{document}
\title{Simultaneous testing of hypotheses and alternatives }
\author{ Koldanov P.A., Koldanov A.P.
\\
NRU HSE, Nizhny Novgorod, Russia  
}
\date{}
\newtheorem{definition}{Definition}[section]
\newtheorem{teo}{Theorem}[section]
\newtheorem{lemma}{Lemma}[section]
\newtheorem{example}{Example}[section]
\newtheorem{corollary}{Corollary}[section]
\newtheorem{note}{Note}[section]
\maketitle
\abstract{}
To identify statistically significant conclusions, it is proposed to simultaneously test hypotheses and alternatives. It is shown that, under the condition of free combination of hypotheses and alternatives, the closure method leads to single-step procedures for the simultaneous testing of hypotheses and alternatives. Using an example it is shown the result is lost when hypotheses and alternatives do not satisfy the condition of free combination of hypotheses and alternatives. It is shown that the average number of insignificant conclusions is an uncontrolled part of the risk function of a single-step procedure with an appropriate choice of the loss function.


\section{Introduction}\label{introduction}

The hypotheses testing theory allows obtaining statistically significant conclusions only in the case of rejecting the tested hypotheses.
If it is desirable to obtain statistically significant conclusions about the truth of the tested hypotheses, it is natural to proceed to testing alternative hypotheses.
When testing hypotheses and alternatives simultaneously, the set of obtained conclusions can be divided into three parts: the set of statistically significant conclusions about the truth of the tested hypotheses, the set of statistically significant conclusions about the falsehood of the tested hypotheses, and the zone of uncertainty (the set of insignificant conclusions), which corresponds to the set of conclusions where both hypotheses and alternatives are simultaneously accepted.
This article investigates the specific of procedures for the simultaneous testing of hypotheses and alternatives that control the probability of at least one false rejection of true hypothesis or true alternative.
The case when hypotheses and alternatives satisfy the condition of free combination is considered.
 
Dividing conclusions into statistically significant and insignificant (or decision rejection) were discussed in early works on simultaneous hypothesis testing \cite{Duncan_1955}, \cite{Lehmann1957_2}. 
Currently, such partitions are of interest in the tasks of constructing and analyzing graphical models \cite{Drton2004}, \cite{Drton2008}, \cite{Koldanov_2023}.
In \cite{Drton2004}, \cite{Drton2008}, a SINful approach to undirected Gaussian graphical model selection is proposed. However, the procedure for identifying the set of statistically significant conclusions about the truth of the hypotheses being tested is heuristic in nature.
The method for dividing the conclusions about the elements of the threshold similarity graph into three parts was proposed in \cite{Koldanov_2023}. The method lead to construction of three sets: a statistically significant set of pairs of nodes connected by an edge; a statistically significant set of pairs of nodes not connected by an edge; and an area of uncertainty. However, the results in \cite{Koldanov_2023} are limited to the simultaneous testing of one-sided hypotheses.

In the article the approach proposed in \cite{Koldanov_2023} is extended to simultaneous testing of hypotheses and alternatives of a more general form. This development allows for distinguishing between statistically significant and insignificant conclusions, considering threshold similarity graph selection as well as undirected Gaussian graphical model selection. Similar considerations apply to covariance graph models, which are defined in terms of marginal independence rather than conditional independence.

The article is organized as follows:  
 In Section \ref{Basic_definitions_and_notations}, the main notations and problem formulation are presented;  
 In Section \ref{mod_proc_Bonf}, the Bonferroni procedure for simultaneous testing of hypotheses and alternatives is discussed;  
 In Section \ref{generalized_closure_principle}, it is shown that, under the condition of free combination of hypotheses and alternatives, applying the basic methods from multiple hypotheses testing theory to the considered class of problems leads to single-step procedures;  
 In Section \ref{characteristics_of_proc}, the relationship between the risk function of a single-step procedure and the expected number of elements in the zone of uncertainty is discussed;  
 In Section \ref{example_free_comb}, the importance of the free combination condition for hypotheses and alternatives is exemplified;  
 In Section \ref{conclusion}, conclusions are presented.
    
\section{Main notations and problem formulation}\label{Basic_definitions_and_notations} 
Let us consider the problem of simultaneous testing of hypotheses $h_i,i=1,\ldots,M$. Let $k_i,i=1,\ldots,M$ be the corresponding alternatives.       
For definiteness, we assume that \( h_i: \theta \in \Theta_i \), \( k_i: \theta \in \Theta_i^c = \Theta \setminus \Theta_i \), where \( \Theta \) is the parameter space.

Suppose that the hypotheses and alternatives \( h_1, k_1, h_2, k_2, \ldots, h_M, k_M \) satisfy the condition of free combination \cite{Liu1996}, i.e., 
\begin{equation}\label{free_combination_condition_strong_sence}
\begin{array}{l}
\forall P\subseteq\{1,2,\ldots,M\} \mbox{ one has } \\
\bigcap_{i\in P}h_i\cap\bigcap_{i\in\{1,2,\ldots,M\}\setminus P}k_i\neq\emptyset\\
\end{array}
\end{equation}
where the intersection is understood as the intersection of the corresponding parameter regions.

Let
$J=\{i; i=1,\ldots,M\}$  be the set of indices of the hypotheses being tested,
$J_t(\theta)$ be the set of indices of the true hypotheses,
$J_f(\theta)$ be the set of indices of the false hypotheses or the true alternatives.

Suppose that the tests \(\varphi_1, \varphi_2, \ldots, \varphi_M; \psi_1, \psi_2, \ldots, \psi_M\) for hypotheses \(h_i\) and alternatives \(k_i\) at an arbitrary level \(0 < \alpha < 1\) are available and have the form:
\begin{equation}\label{test_for_hypotheses}
\varphi_i=\left\{\
									\begin{array}{ll}
									1,& \hat{p}_{h_i}<\alpha\\
									0,& \hat{p}_{h_i}\geq\alpha
									\end{array}
							 \right.,
\end{equation}
\begin{equation}\label{test_for_alternatives}
\psi_i=\left\{\
									\begin{array}{ll}
									1,& \hat{p}_{k_i}<\alpha\\
									0,& \hat{p}_{k_i}\geq\alpha
									\end{array}
							 \right.
\end{equation}
where $\hat{p}_{h_i},\hat{p}_{k_i}$ are p-values of hypotheses $h_i$, alternatives $k_i$ respectively.
Note that if $[\Theta_i]\cap[\Theta_i^c]\neq\emptyset$ for all $i=1,\ldots,M$, where $[\Theta_i]$ is the closure of $\Theta_i$, then $\hat{p}_{h_i}+\hat{p}_{k_i}=1$.
 
Let's denote:

- \( U = \{ i : \varphi_i = 0 \} \) - the set of indices of hypotheses \( h_i \) that are accepted,

- \( \overline{U} = \{ i : \varphi_i = 1 \} \) - the set of indices of hypotheses \( h_i \) that are rejected,

- \( L = \{ i : \psi_i = 1 \} \) - the set of indices of alternatives \( k_i \) that are rejected,

- \( \overline{L} = \{ i : \psi_i = 0 \} \) - the set of indices of alternatives \( k_i \) that are accepted,

- \( |A| \) - the number of elements in the set \( A \).

Let's consider the problem of constructing a procedure \(\delta = (\varphi_1, \varphi_2, \ldots, \varphi_M; \psi_1, \psi_2, \ldots, \psi_M)\) for simultaneous testing of hypotheses \(h_i\) and alternatives \(k_i\), which controls the Family-Wise Error Rate (FWER) at level \(\alpha\) in the strong sense and satisfies the condition
\begin{equation}\label{sets_inclusion}
L\subseteq U.
\end{equation}
It is evident that, to satisfy (\ref{sets_inclusion}), it is sufficient that \(\hat{p}_{h_i} + \hat{p}_{k_i} = 1\) for all \(i=1,\ldots,M\).

If such a procedure \(\delta\) is constructed, then the sets \(\overline{U}\) and \(L\) are the sets of statistically significant conclusions, and the set \(G = U \setminus L\) is the set of simultaneously accepted hypotheses and alternatives, which can naturally be called the set of insignificant conclusions or the zone of uncertainty.
 
%
 
\section{Bonferroni procedure}\label{mod_proc_Bonf}

For the control of directional errors, in \cite{Bauer_1986} a procedure was proposed, which can be considered as simultaneous testing procedure for one-sided hypotheses \(h_i: \theta_i \leq 0\) and alternatives \(k_i: \theta_i > 0,\ i=1,\ldots,M\), that controls the Family-Wise Error Rate (FWER) at level \(\alpha\) in the strong sense. Such procedure has the form:
\begin{equation}\label{modified_bonf_Bauer}
\delta=(\varphi_1,\psi_1,\ldots,\varphi_M,\psi_M)
\end{equation}
$$\varphi_i=\left\{\ 
							\begin{array}{ll}
							1,& \hat{p}_i^{\leq}<\frac{\alpha}{M}\\
							0,& \hat{p}_i^{\leq}\geq\frac{\alpha}{M}
							\end{array},
						\right.
	\psi_i=\left\{\ 
							\begin{array}{ll}
							1,& \hat{p}_i^{\geq}<\frac{\alpha}{M}\\
							0,& \hat{p}_i^{\geq}\geq\frac{\alpha}{M}
							\end{array}
				\right.,
$$ 
where $\hat{p}_i^{\leq},\hat{p}_i^{\geq}$ are p-values of $h_i,k_i$ respectively.

Note that when \(\frac{\alpha}{M} \leq \frac{1}{2}\) (i.e., for \(M \geq 2\)), the inclusion \(L \subseteq U\) holds.

Let's consider arbitrary hypotheses \(h_j, j=1,\ldots,M\), and alternatives \(k_j\). Since
$$P_{\theta}(\exists j\in L\cap J_f\mbox{ or }\exists k\in J_t\cap\overline{U}) =P_{\theta}(\exists j\in L\cap J_f)+P_{\theta}(\exists k\in J_t\cap\overline{U})\leq$$
$$\leq \frac{|J_t|\alpha}{M}+\frac{|J_f|\alpha}{M}=\alpha,\ \  \forall \theta\in\Theta,$$
then the procedure \(\delta\), using tests (\ref{test_for_hypotheses}) and (\ref{test_for_alternatives}), controls the \(FWER \leq \alpha\) in the strong sense in this case.

Therefore if
$$\hat{p}_{h_i}+\hat{p}_{k_i}=1$$
then a straightforward generalization of the procedure (\ref{modified_bonf_Bauer}), which uses tests (\ref{test_for_hypotheses}) and (\ref{test_for_alternatives}), leads to the construction of the sets of statistically significant conclusions \(\overline{U}(x) = J \setminus U(x)\) and \(L(x)\), which have the form:
\begin{equation}\label{sets_mod_proc_Bonf}
\begin{array}{l}
\overline{U}(x)=\{i:\hat{p}_{h_{i}}\leq\frac{\alpha}{M}\}\\
L(x)=\{i:\hat{p}_{h_{i}}>1-\frac{\alpha}{M}\}.
\end{array}
\end{equation}

\section{Closure method}\label{generalized_closure_principle}
When applying the closure method, it is sufficient to test only those intersection hypotheses \(H_Q = \bigcap_{i \in Q} h_i\) that are non-empty.  
In the case of simultaneous testing of the hypotheses and the alternatives, the intersection hypotheses \(H_{J_1,J_2} = \left(\bigcap_{i \in J_1} h_i\right) \cap \left(\bigcap_{i \in J_2} k_i\right)\), where \(J_1 \cap J_2 \neq \emptyset\), are empty.

If the hypotheses and the alternatives being tested satisfy the condition of free combination, then among all non-empty intersection hypotheses, the minimal ones \cite{Finner_2002} are the intersection hypotheses of the form:$$H_{J_1}=H_{J_1,\{1,\ldots,M\}\setminus J_1}=\bigcap_{i\in J_1}h_i\cap\bigcap_{i\in \{1,\ldots,M\}\setminus J_1}k_i.$$
Let $\{1,2,\ldots,M\}\setminus(J_1\cup J_2)=\{i_1,\ldots,i_k\}$. It is evident  
$$H_{J_1,J_2}=\bigcap_{i\in J_1}h_i\cap\bigcap_{i\in J_2}k_i=\bigcup_{\pi\subseteq\{i_1,\ldots,i_k\}}H_{J_1\cup\pi}=\bigcup_{\pi\subseteq\{i_1,\ldots,i_k\}}\left(\bigcap_{i\in J_1\cup\pi}h_i\cap\bigcap_{i\in \{1,\ldots,M\}\setminus \{J_1\cup\pi\}}k_i\right)$$
Thus, under the condition of free combination of hypotheses and alternatives, the hypotheses \(H_{J_1,J_2}\) are non-empty for all \(J_1, J_2 \subset \{1,\ldots,M\}\) such that \(J_1 \cap J_2 = \emptyset\).
    
\begin{teo}\label{closure_equivalent_bonferroni_theorem}
If
\begin{enumerate}
\item The hypotheses \(h_i\) and the alternatives \(k_i\) satisfy the condition of free combination,
\item Tests of the union-intersection type are used to evaluate the intersection hypotheses \(H_{J_1,J_2}\)\cite{Roy1952},
\item The condition (\ref{sets_inclusion}) is satisfied,
\end{enumerate}
then the closure method leads to a single-step procedure for the simultaneous testing of hypotheses and alternatives.
\end{teo} 
\begin{proof}
Without loss of generality consider hypothesis $h_1$.  
Hypothesis $h_1$ is rejected by closure method \cite{Marcus1976} if and only if all non-empty hypotheses $H_{J_1,J_2}$, such that $\{1\}\subseteq J_1$, are rejected. Let $J_1=\{1,i_2,\ldots,i_{|J_1|}\}, J_2=\{j_1,\ldots,j_{|J_2|}\}, J_1\cap J_2=\emptyset$. 

Tests of the union-intersection type used for \(H_{J_1,J_2}\) have the form \cite{Roy1952}:
$$
\delta_{J_1,J_2}=\left\{\
									\begin{array}{ll}
									1,& \min_{i\in J_1,j\in J_2}(\hat{p}_{h_i},\hat{p}_{k_j})<\alpha(|J_1\cup J_2|)\\
									0,&\mbox{ otherwise, }
									\end{array}
								 \right.
$$
where $\alpha(|J_1\cup J_2|)$ is defined from
\begin{equation}\label{union_intersection_test_hypotheses_alternative_L}
\begin{array}{l}
P_{\theta}(\min_{i\in J_1,j\in J_2}(\hat{p}_{h_i},\hat{p}_{k_j})<\alpha(|J_1\cup J_2|))\leq\alpha\\
\end{array}
\end{equation}
$$\mbox{ under }\theta\in \bigcup_{i\in J_1\cup J_2}\left([ \Theta_i ] \cap [\Theta_i^c ]\right).$$

When \(|J_1 \cup J_2| > 1\) and for any index \(m \in \{2, \ldots, M\}\), the hypotheses set \(H_{J_1,J_2} : \{1\} \subset J_1\) includes both intersection hypotheses containing \(h_m\) and intersection hypotheses containing \(k_m\), with \(m \neq 1\) and \(m \in J_1 \cup J_2\), where \(J_1 \cap J_2 = \emptyset\).  
According to condition 3 of the theorem \(\hat{p}_{h_m}\) and \(\hat{p}_{k_m}\) cannot both be less than \(\alpha(|J_1 \cup J_2|)\).  
Therefore, all hypotheses \(H_{J_1,J_2}\) can be rejected simultaneously if and only if \(\hat{p}_{h_1} < \alpha(|J_1 \cup J_2|)\).  
Since \(|J_1 \cup J_2| \leq M\), it follows that \(\min_{|J_1 \cup J_2|} \alpha(|J_1 \cup J_2|) = \alpha(M)\).
 
Therefore, the test for \(h_1\) has the form:
$$\varphi_1=\left\{\
									\begin{array}{ll}
									1,& \hat{p}_{h_1}<\alpha(M)\\
									0,& \hat{p}_{h_1}\geq\alpha(M)
									\end{array}
							 \right.,
$$
Then the procedure for the simultaneous testing of hypotheses \(h_1, \ldots, h_M\) and alternatives \(k_1, \ldots, k_M\) is a single-step procedure and has the form:
\begin{equation}\label{single_step_proc_test_hypot_alter}
\left(\varphi_1,\ldots,\varphi_M,\psi_1,\ldots,\psi_M\right)
\end{equation}
where 
$$\varphi_i=\left\{\
									\begin{array}{ll}
									1,& \hat{p}_{h_i}<\alpha(M)\\
									0,& \hat{p}_{h_i}\geq\alpha(M)
									\end{array}
							 \right.,
$$
$$\psi_i=\left\{\
									\begin{array}{ll}
									1,& \hat{p}_{k_i}<\alpha(M)\\
									0,& \hat{p}_{k_i}\geq\alpha(M)
									\end{array}
							 \right.,
$$
\begin{equation}\label{union_intersection_test_hypotheses_alternative_M_equivalent}
P_{\theta}(\min_{i\in\{1,\ldots,M\}}(\hat{p}_{h_i})<\alpha(M))\leq\alpha
\end{equation}

\end{proof}
\begin{note}
If \(\hat{p}_{h_i}\) are independent for all \(i=1,\ldots,M\), then \(\alpha(M) = 1 - \sqrt[M]{1 - \alpha}\).  
If, for testing intersection hypotheses, the Bonferroni procedure is used, then \(\alpha(M) = \frac{\alpha}{M}\).
\end{note}
Similar reasoning shows that the partitioning method \cite{Finner_2002} leads to a single-step procedure.  
Moreover, from \cite{Shaffer_1986} follows that the modified Holm procedure for the simultaneous testing of hypotheses \(h_1, \ldots, h_M\) and alternatives \(k_1, \ldots, k_M\) is also a single-step procedure.  
Essentially, these results are a consequence of the fact that out of the \(2M\) hypotheses and alternatives being tested - satisfying the condition of free combination - exactly \(M\) are true.

\section{Characteristics of procedures for simultaneous testing of hypotheses and alternatives}\label{characteristics_of_proc}
Under $\hat{p}_{h_i}+\hat{p}_{k_i}=1$ the procedure (\ref{single_step_proc_test_hypot_alter}), (\ref{union_intersection_test_hypotheses_alternative_M_equivalent}) could be written in the following form: 
\begin{equation}\label{three_decision_procedure}
\delta=(\delta_1,\ldots,\delta_M)
\end{equation}
where
$$
\delta_i=\left\{\
					 \begin{array}{ll}
							d_1,&\mbox{ if }\hat{p}_{h_i}>1-\alpha(M)\\
							d_2,&\mbox{ if }\hat{p}_{h_i}<\alpha(M)\\
							d_3,&\mbox{ if }\alpha(M)\leq \hat{p}_{h_i}\leq 1-\alpha(M),\\
					 \end{array}
				\right.	
$$
\(d_1\) - a significant decision that the alternative \(k_i\) is false, i.e., the hypothesis \(h_i\) is true,

\noindent\(d_2\) - a significant decision that the hypothesis $h_i$  is false, i.e. the alternative \(k_i\) is false,

\noindent\(d_3\) - a decision to accept both the hypothesis \(h_i\) and the alternative \(k_i\), for \(i=1,\ldots,M\).

In the case
$$L=\{i:\delta_i=d_1\},$$
$$\overline{U}=\{i:\delta_i=d_2\},$$
$$G=U\setminus L=\{i:\delta_i=d_3\}.$$
As a characteristic of the procedure (\ref{three_decision_procedure}) for simultaneous testing of the hypotheses and the alternatives, it is natural to use \(E(|G|)\) \cite{Koldanov_2023}.

Let's show that, with an appropriate choice of a loss function, \(E(|G|)\) is closely related to the risk function \cite{Lehmann1957_1}, \cite{Lehmann1957_2}.  
Denote by \(a_i, b_i\), for \(i=1,\ldots,M\), the losses associated with incorrect rejection and acceptance of the hypothesis \(h_i\); and by \(c_i, l_i\), for \(i=1,\ldots,M\), the losses from incorrect rejection and acceptance of the alternative \(k_i\).

For a fixed \(i\) define the loss function \(w(\theta, \delta_i)\) as follows:
\begin{equation}\label{loss_function}
w(\theta, \delta_i)=\left\{\
											\begin{array}{ll}
											 0,&\theta\in h_i,\delta_i=d_1\\
											 a_i+l_i,& \theta\in h_i,\delta_i=d_2\\
											 l_i,& \theta\in h_i,\delta_i=d_3\\
											 c_i+b_i,& \theta\in k_i,\delta_i=d_1\\
										   0,&\theta\in k_i,\delta_i=d_2\\
							         b_i,& \theta\in k_i,\delta_i=d_3\\
											\end{array}
										\right.	
\end{equation}

Suppose that \( a_i = c_i = a \), \( b_i = l_i = b \), and \( a + b = 1 \) for all \( i=1,\ldots,M \). Moreover, assume that the loss function is additive \cite{Lehmann1957_1}, \cite{Lehmann1957_2}, i.e.,
$$w(\theta,\delta)=\sum_{i=1}^Mw(\theta,\delta_i).$$
Then
$$R(\theta,\delta)=E_{\theta}w(\theta,\delta)=(a+b)\sum_{i:\theta\in h_i}P_{\theta}(\delta_i=d_2)+(a+b)\sum_{i:\theta\in k_i}P_{\theta}(\delta_i=d_1)+b\sum_{i=1}^MP_{\theta}(\delta_i=d_3)=$$
$$=\sum_{i:\theta\in h_i}P_{\theta}(\delta_i=d_2)+\sum_{i:\theta\in k_i}P_{\theta}(\delta_i=d_1)+b\sum_{i=1}^MP_{\theta}(\delta_i=d_3)$$
The third term \( R(\theta, \delta) \), up to \( b \), coincides with \( E(|G|) \). If the procedure \( (\varphi_1, \psi_1, \ldots, \varphi_M, \psi_M) \) controls the Family-Wise Error Rate (FWER) at level \( \leq \alpha \), then the sum of the first two terms of \( R(\theta, \delta) \) does not exceed \( \alpha \). Then the uncontrolled component of the risk function, under the loss function given by (\ref{loss_function}), corresponds to the expectation of the number of elements in the uncertainty zone.

\section{Counterexample}\label{example_free_comb}
Let $x_1,\ldots,x_n$ be a sample from $N(\theta,1)$.
Given \( \theta_1 < \theta_2 \), consider the problem of constructing a procedure \( \delta = (\varphi_1, \psi_1, \varphi_2, \psi_2) \) for simultaneous testing of the hypotheses \( h_i: \theta \geq \theta_i \), \( i=1,2 \), against the alternatives \( k_i: \theta < \theta_i \).

Since \( h_2 \cap k_1 = \emptyset \), the hypotheses \( h_i: \theta \geq \theta_i \), \( i=1,2 \), and the alternatives \( k_i: \theta < \theta_i \) obviously do not satisfy the condition of free combination of hypotheses and alternatives.

Standard tests for hypotheses \( h_i: \theta \geq \theta_i \) have the form:
$$
\varphi_{h_i}=\left\{\
												\begin{array}{ll}
												 1,&\sqrt{n}(\overline{x}-\theta_i)<-c_{\alpha}\\
												 0,&\sqrt{n}(\overline{x}-\theta_i)\geq -c_{\alpha}		
												\end{array}
											\right.	
$$
$-c_{\alpha}$ - level $\alpha$ quantile of the standard normal distribution $N(0,1)$.
 
Standard tests for alternatives $k_i:\theta<\theta_i$ have the form:
$$
\varphi_{k_i}=\left\{\
												\begin{array}{ll}
												 1,&\sqrt{n}(\overline{x}-\theta_i)>c_{\alpha}\\
												 0,&\sqrt{n}(\overline{x}-\theta_i)\leq c_{\alpha}		
												\end{array}
											\right.	
$$
$c_{\alpha}$ - percentile of the standard normal distribution $N(0,1)$. 

Construction of Bonferroni procedure is straightforward.

The tests \( \varphi_1, \psi_1, \varphi_2, \psi_2 \) for \( h_1, k_1, h_2, k_2 \) of the procedure \( \delta = (\varphi_1, \psi_1, \varphi_2, \psi_2) \), constructed using the closure method, have the form:
$$
\varphi_{1}=\left\{\
												\begin{array}{ll}
												 1,&\sqrt{n}(\overline{x}-\theta_1)<-c_{\frac{\alpha}{2}}\\
												 0,&\sqrt{n}(\overline{x}-\theta_1)\geq -c_{\frac{\alpha}{2}}		
												\end{array}
											\right.	
$$
$$
\psi_{1}=\left\{\
												\begin{array}{ll}
												 1,&\sqrt{n}(\overline{x}-\theta_1)>c_{\alpha}\\
												 0,&\sqrt{n}(\overline{x}-\theta_1)\leq c_{\alpha}		
												\end{array}
											\right.	
$$
$$
\varphi_{2}=\left\{\
												\begin{array}{ll}
												 1,&\sqrt{n}(\overline{x}-\theta_2)<-c_{\alpha}\\
												 0,&\sqrt{n}(\overline{x}-\theta_2)\geq -c_{\alpha}		
												\end{array}
											\right.	
$$
$$
\psi_{2}=\left\{\
												\begin{array}{ll}
												 1,&\sqrt{n}(\overline{x}-\theta_2)>c_{\frac{\alpha}{2}}\\
												 0,&\sqrt{n}(\overline{x}-\theta_2)\leq c_{\frac{\alpha}{2}}		
												\end{array}
											\right.	
$$

Thus, the procedure \( \delta = (\varphi_1, \psi_1, \varphi_2, \psi_2) \), constructed using the closure method, do not coincide with Bonferroni procedure.

\section{Conclusion}\label{conclusion}

Specific of procedures for simultaneous testing of hypotheses and alternatives is investigated. Such procedures allow for distinguishing between significant and insignificant conclusions about the hypotheses being tested. The idea of constructing such procedures for analyzing a certain class of graphical models was proposed in \cite{Koldanov_2023}. However, the results obtained in \cite{Koldanov_2023} was limited to testing one-sided hypotheses and application of the  Bonferroni procedure, which in \cite{Romano_2016} is referred to as a modified Bonferroni procedure. 

This article develops the results obtained in \cite{Koldanov_2023}. It is proven that, under the condition of free combination of hypotheses and alternatives, the closure method leads to single-step procedures. Using an example it is shown the result is lost when
hypotheses and alternatives do not satisfy the condition of free combination
of hypotheses and alternatives. It is shown that the average number of insignificant conclusions can be considered as an uncontrolled part of the risk function of a single-step procedure. This allows for applying the results on constructing optimal procedures, obtained in \cite{Lehmann1957_1} and \cite{Lehmann1957_2}.

\end{document}